\documentclass[12pt]{article}
\usepackage[utf8]{inputenc}
\usepackage[a4paper, margin=1in]{geometry}
\usepackage{mathtools}
\mathtoolsset{showonlyrefs=true}
\usepackage{amsthm, amssymb, bbm, mathrsfs}
\usepackage{bm}
\usepackage{graphicx}
\usepackage{float}
\usepackage[colorlinks=true, linkcolor=blue, urlcolor=cyan, citecolor=green]{hyperref}
\usepackage[dvipsnames]{xcolor}
\usepackage{algorithm2e}
\theoremstyle{plain}
\newtheorem{theorem}{Theorem}[section]

\newtheorem{lemma}[theorem]{Lemma}
\newtheorem{remark}[theorem]{Remark}

\newtheorem{cor}[theorem]{Corollary}

\newtheorem{prop}[theorem]{Proposition}
\numberwithin{equation}{section}
\numberwithin{figure}{section}

\newcommand{\E}{\mathbb{E}}

\newcommand{\p}{\partial}

\newcommand{\beas}{\begin{eqnarray*}}
\newcommand{\eeas}{\end{eqnarray*}}
\newcommand{\bal}{\begin{align}}
\newcommand{\eal}{\end{align}}
\newcommand{\bas}{\begin{align*}}
\newcommand{\eas}{\end{align*}}
\newcommand{\bea}{\begin{eqnarray}}
\newcommand{\eea}{\end{eqnarray}}
\newcommand{\tmop}{\end{eqnarray}}
\newcommand{\ben}{\begin{enumerate}}
\newcommand{\een}{\end{enumerate}}

\newcommand{\ui}{\mathrm{i}}

\newcommand{\dm}{\diamond}

\newcommand{\cB}{\mathcal{B}}

\newcommand{\cL}{\mathcal{L}}
\newcommand{\cF}{\mathcal{F}}

\newcommand{\mF}{\mathbb{F}}

\newcommand{\cG}{\mathcal{G}}

\newcommand{\He}{\mathrm{He}}

\newcommand{\tmF}{\tilde {\mathbb{F}}}

\newcommand{\cO}{\mathcal{O}}

\newcommand{\bi}{\begin{itemize}}
\newcommand{\ei}{\end{itemize}}
\newcommand{\beq}{\begin{equation}}
\newcommand{\eeq}{\end{equation}}
\newcommand{\bv}{\begin{verbatim}}
\newcommand{\ev}{\end{verbatim}}

\newcommand{\ee}[1]{ {\mathbb{E}\left[{#1}\right]}}

\usepackage{tikz}
\newcommand{\tkz}{\tikzexternaldisable}
\usepackage{mhequ} 
\usetikzlibrary{snakes}
\usetikzlibrary{decorations}
\usetikzlibrary{positioning}
\usetikzlibrary{shapes}
\usetikzlibrary{external}
\makeatletter
\pgfdeclareshape{crosscircle}
{
  \inheritsavedanchors[from=circle] 
  \inheritanchorborder[from=circle]
  \inheritanchor[from=circle]{north}
  \inheritanchor[from=circle]{north west}
  \inheritanchor[from=circle]{north east}
  \inheritanchor[from=circle]{center}
  \inheritanchor[from=circle]{west}
  \inheritanchor[from=circle]{east}
  \inheritanchor[from=circle]{mid}
  \inheritanchor[from=circle]{mid west}
  \inheritanchor[from=circle]{mid east}
  \inheritanchor[from=circle]{base}
  \inheritanchor[from=circle]{base west}
  \inheritanchor[from=circle]{base east}
  \inheritanchor[from=circle]{south}
  \inheritanchor[from=circle]{south west}
  \inheritanchor[from=circle]{south east}
  \inheritbackgroundpath[from=circle]
  \foregroundpath{
    \centerpoint%
    \pgf@xc=\pgf@x%
    \pgf@yc=\pgf@y%
    \pgfutil@tempdima=\radius%
    \pgfmathsetlength{\pgf@xb}{\pgfkeysvalueof{/pgf/outer xsep}}%
    \pgfmathsetlength{\pgf@yb}{\pgfkeysvalueof{/pgf/outer ysep}}%
    \ifdim\pgf@xb<\pgf@yb%
      \advance\pgfutil@tempdima by-\pgf@yb%
    \else%
      \advance\pgfutil@tempdima by-\pgf@xb%
    \fi%
    \pgfpathmoveto{\pgfpointadd{\pgfqpoint{\pgf@xc}{\pgf@yc}}{\pgfqpoint{-0.707107\pgfutil@tempdima}{-0.707107\pgfutil@tempdima}}}
    \pgfpathlineto{\pgfpointadd{\pgfqpoint{\pgf@xc}{\pgf@yc}}{\pgfqpoint{0.707107\pgfutil@tempdima}{0.707107\pgfutil@tempdima}}}
    \pgfpathmoveto{\pgfpointadd{\pgfqpoint{\pgf@xc}{\pgf@yc}}{\pgfqpoint{-0.707107\pgfutil@tempdima}{0.707107\pgfutil@tempdima}}}
    \pgfpathlineto{\pgfpointadd{\pgfqpoint{\pgf@xc}{\pgf@yc}}{\pgfqpoint{0.707107\pgfutil@tempdima}{-0.707107\pgfutil@tempdima}}}
  }
}
\makeatother

\def\X{\tikz[baseline=-2.8,scale=0.15]{\node[X] {};}} 
\def\M{\tikz[baseline=-2.8,scale=0.15]{\node[M] {};}} 

\def\MXd{\tikz[baseline=-1,scale=0.15]{\draw (-1,1) node[M] {} -- (0,0) node[not] {} -- (1,1) node[X] {};}} 

\def\MMd{\tikz[baseline=-1,scale=0.15]{\draw (-1,1) node[M] {} -- (0,0) node[not] {} -- (1,1) node[M] {};}} 

\def\MXdXd{\tikz[baseline=-1,scale=0.15]{
\draw (0,0) node[not] {} -- (-1,1) node[not] {}
-- (-2,2) node[M]{} ;
\draw (0,0) -- (1,1) node[X] {};
\draw (-1,1) -- (0,2) node[X] {};
}}

\def\MMdXd{\tikz[baseline=-1,scale=0.15]{
\draw (0,0) node[not] {} -- (-1,1) node[not] {}
-- (-2,2) node[M]{} ;
\draw (0,0) -- (1,1) node[X] {};
\draw (-1,1) -- (0,2) node[M] {};
}}

\def\MXdXdXd{\tikz[baseline=-1,scale=0.15]{
\draw (0,0) node[not] {} -- (-1,1) node[not] {}
-- (-2,2) node[not]{}  -- (-3,3) node[M]{};
\draw (0,0) -- (1,1) node[X] {};
\draw (-1,1) -- (0,2) node[X] {};
\draw (-2,2) -- (-1,3) node[X] {};
}}

\def\MXdMd{\tikz[baseline=-1,scale=0.15]{
\draw (0,0) node[not] {} -- (-1,1) node[not] {}
-- (-2,2) node[M]{} ;
\draw (-1,1) -- (0,2) node[X] {};
\draw (0,0) -- (1,1) node[M] {};
}}

\def\MMdMd{\tikz[baseline=1,scale=0.15]{
\draw (0,0) node[not] {} -- (-1,1) node[not] {}
-- (-2,2) node[M]{} ;
\draw (0,0) -- (1,1) node[M] {};
\draw (-1,1) -- (0,2) node[M] {};
}}

\def\MMdXdXd{\tikz[baseline=1,scale=0.15]{
\draw (0,0) node[not] {} -- (-1,1) node[not] {}
-- (-2,2) node[not]{}  -- (-3,3) node[M]{};
\draw (-2,2) -- (-1,3) node[M] {};
\draw (-1,1) -- (0,2) node[X] {};
\draw (0,0) -- (1,1) node[X] {};
}}

\def\MXdMdXd{\tikz[baseline=1,scale=0.15]{
\draw (0,0) node[not] {} -- (-1,1) node[not] {}
-- (-2,2) node[not]{}  -- (-3,3) node[M]{};
\draw (-2,2) -- (-1,3) node[X] {};
\draw (-1,1) -- (0,2) node[M] {};
\draw (0,0) -- (1,1) node[X] {};
}}

\def\MXdXdMd{\tikz[baseline=1,scale=0.15]{
\draw (0,0) node[not] {} -- (-1,1) node[not] {}
-- (-2,2) node[not]{}  -- (-3,3) node[M]{};
\draw (-2,2) -- (-1,3) node[X] {};
\draw (-1,1) -- (0,2) node[X] {};
\draw (0,0) -- (1,1) node[M] {};
}}

\def\MXdMXdd{\tikz[baseline=1,scale=0.15]{\draw (0,0) node[not] {} -- (-1,1) node[not] {};
\draw (0,0) -- (1,1) node[not] {};
\draw (-1,1) -- (-1.5,2.5) node[M] {};
\draw (-1,1) -- (-0.5,2.5) node[X] {};
\draw (1,1) -- (0.5,2.5) node[M] {};
\draw (1,1) -- (1.5,2.5) node[X] {};}}

\def\MXdXdXdXd{\tikz[baseline=1,scale=0.15]{
\draw (0,0) node[not] {} -- (-1,1) node[not] {}
-- (-2,2) node[not]{} -- (-3,3) node[not]  {} 
-- (-4,4) node[M]{};
\draw (-3,3) -- (-2,4) node[X] {};
\draw (-2,2) -- (-1,3) node[X] {};
\draw (-1,1) -- (0,2) node[X] {};
\draw (0,0) -- (1,1) node[X] {};
}}

\colorlet{symbols}{blue!90!black}
\colorlet{testcolor}{green!60!black}
\colorlet{connection}{red!30!black}

\tikzset{
root/.style={circle,fill=black!50,inner sep=0pt, minimum size=3mm},
        dot/.style={circle,fill=black,inner sep=0pt, minimum size=1.5mm},
        A/.style={very thin,circle,fill=BlueGreen!100,draw=black,inner sep=0pt,minimum size=1.2mm}, 
        X/.style={very thin,circle,fill=Gray!80,draw=black,inner sep=0pt,minimum size=1.2mm}, 
        Z/.style={very thin,circle,fill=RubineRed!100,draw=black,inner sep=0pt,minimum size=1.2mm}, 
        M/.style={very thin,circle,fill=YellowOrange!100,draw=black,inner sep=0pt,minimum size=1.2mm}, 
not/.style={thin,circle,fill=symbols,draw=connection,fill=connection,inner sep=0pt,minimum size=0.35mm},
	>=stealth,
        }

\newcommand{\btau}{\bm{\tau}}
\begin{document}
\title{\textbf{Demystifying the Bergomi–Guyon expansion}}
\author{Florian Bourgey\footnote{NYU, fb2615@nyu.edu}, Jim Gatheral\footnote{Baruch College, CUNY, jim.gatheral@baruch.cuny.edu}}
\date{\today}
\maketitle

\begin{abstract}
Al\`os, Gatheral and Radoi\v{c}i\'c derived the Bergomi--Guyon
expansion of the implied variance smile from the forest expansion of the
cumulant generating function. Its coefficients are sums of products of
diamond trees, with prefactors that are polynomials in the log-strike
$k$. Matching moments order by order produces, at order $\epsilon^\ell$,
terms of degree greater than $\ell$ in $k$ that mysteriously  cancel.

We show that, in suitable variables, the matching condition can be formulated as a
nonlinear heat equation. The resulting recursion computes the prefactor
of each product of trees from those of products of fewer trees, without
generating the higher-degree terms; its only model-independent input is
a cumulant series, which we compute in closed form. Consequently,
at order $\epsilon^\ell$ every prefactor has degree exactly $\ell$ in
$k$. The prefactors are universal and need only be computed once.
Code and coefficients are provided.
\end{abstract}




\section{Introduction}\label{sec:intro}
Lewis \cite{lewis2000option}, and subsequently Bergomi and Guyon
\cite{bergomi2012stochastic} (BG), compute option prices to second order in
volatility of volatility and then invert them to obtain an expansion of
the implied volatility smile. Both find that the second-order
coefficient is quadratic in the log-strike $k$, even though the
intermediate computation contains cubic and quartic terms.
Jacquier and Lorig \cite{jacquier2015characteristic} bypass the option
price by passing algorithmically from the characteristic function directly to the smile. Al\`os, Gatheral and Radoi\v{c}i\'c
\cite{alos2020exponentiation} (AGR) likewise start from the cumulant
generating function (CGF) and expand it in forests of diamond trees, thereby
associating smile coefficients directly with trees. This forest expansion
was subsequently generalized by Friz, Gatheral and Radoi\v{c}i\'c
\cite{friz2022forests}.
In the AGR computation, the same phenomenon reappears
\cite[Remark A.1]{alos2020exponentiation}: ``we again observe the
mysterious cancellation of coefficients of $k^3$ and $k^4$.''

In this note, we replace the computation that produces the cancellation
by one in which these higher-degree terms never arise. The passage from
the cumulant generating function to the smile is governed by a single
operator that, in the right variables, generates Hermite polynomials. In
those variables, the matching condition becomes a nonlinear backward heat
equation and yields a recursion that determines the prefactor of each
product of trees from the prefactors of products of fewer trees. Its only
model-independent input is the cumulant series of that operator, which we
compute in closed form. This recursion is stated in Theorem
\ref{thm:recursion}, the main result of the note; the moments in which
the cancellation occurs never appear in the recursion. Passing from those
moments to their cumulants performs these cancellations once and for all.

Section \ref{sec:example} revisits the second-order example of
\cite[eq.~(A4)]{alos2020exponentiation} using both methods. In the
recursion, no quartic term appears; in the direct computation, the
quartic terms appear and cancel.

Proposition \ref{prop:degree} and Corollary \ref{cor:degree} show that
this pattern repeats at every order: the prefactor of a product of
$j$ trees of total order $\ell$ is a polynomial in the log-strike $k$ of
degree exactly $\ell$, rather than the $\ell+2(j-1)$ degree that a
direct computation would produce.

Appendix \ref{app:BG4} displays the expansion through order four;  coefficients through order six, together with code implementing the algorithm and a tutorial,
are available in the
GitHub repository
\url{https://github.com/fbourgey/bergomi-guyon/}.
\section{The forest expansion and the smile}\label{sec:setup}
We work with forward variance models of the form
\begin{align}
\frac{dS_t}{S_t} &= \sqrt{V_t}\left(\rho\, dW_t + \sqrt{1-\rho^2}\, dW^\perp_t\right),\notag\\
d\xi_t(u) &= f_t(\xi)\, g(u-t)\, dW_t,
\label{eq:SVmodel}
\end{align}
where $\E_t[\cdot] = \E[\cdot|\cF_t]$, $X = \log S$, $\xi_t(u):=\E_t[V_u]$ for $t\leq u$, so that $\xi_t(t)=V_t$, and $V_t\,dt=d\langle X\rangle_t$. Here, $f_t(\xi)$ is an adapted scalar functional of the forward variance curve,
and $g$ is a deterministic kernel. We set
$X_{t,T}:=X_T-X_t$. The Brownian
motions $W$ and $W^\perp$ are independent, and $\rho\in[-1,1]$. For two
continuous semimartingales $A$ and $B$, the diamond product is
\[
(A\dm B)_t(T):=\E_t\left[\langle A,B\rangle_T\right]-\langle A,B\rangle_t,
\]
and
\[
M_t(T):=\E_t\left[\langle X\rangle_T-\langle X\rangle_t\right]
=\int_t^T\xi_t(u)\,du
\]
is the variance contract. When $t=0$, we drop the dependence on $t$ and
write $(A\dm B)(T):=(A\dm B)_0(T)$ and $M(T):=M_0(T)$. We also suppress
the time and maturity arguments in the tree notation. 
From \cite[Corollary 3.1]{alos2020exponentiation},
the CGF has the forest expansion
\beq
\psi(T;u) := \log \ee{e^{\ui u X_{0,T}}}
= - \frac{1}{2}u\,(u+\ui)\,M(T) + \sum_{\ell=1}^\infty\,\tilde{\mF}_\ell(u),
\label{eq:CGFmF}
\eeq
\tkz
where $\tmF_0=-\tfrac12 u(u+\ui)\,\M$ and, for $\ell>0$,
\beq
\tmF_\ell=\frac12\,\sum_{j=0}^{\ell-2}\,\bigl(\tmF_{\ell-2-j} \dm \tmF_j\bigr)
+\ui\, u\, \bigl(X \dm \tmF_{\ell-1}\bigr).
\label{eq:tmFrecursion}
\eeq
The diamond product is commutative but not associative, and is
represented graphically by \emph{root joining}: $\tau_1\dm\tau_2$ is the
tree whose root carries the two subtrees $\tau_1$ and $\tau_2$
\cite[Section 1.2]{friz2022forests}. Iterating \eqref{eq:tmFrecursion}
from $\tmF_0$ therefore generates binary trees whose leaves are labelled
either \M{}, standing for the variance contract $M$, or \X{}, standing
for the log-price $X$; thus $\MXd = M\dm X$, $\MMd = M\dm M$, and
$\MXdXd = (M\dm X)\dm X$. Since $\dm$ is not associative, the bracketing
matters: a tree records not only which leaves occur, but how they are
paired.

Each leaf \M{} contributes a factor $-\tfrac12 u(u+\ui)$ and each leaf
\X{} a factor $\ui u$. The matching condition \eqref{eq:formalEqn} below
involves the CGF not at $u$ but at the shifted argument $u-\ui/2$, the
shift that symmetrizes the Fourier integral; there these two factors
become $-\tfrac12(u^2+\tfrac14)$ and $\ui u+\tfrac12$, which are the
polynomials
\[
\lambda_a:=\tfrac12\,a\,(a-1)
\qquad\text{and}\qquad a ,
\]
evaluated at $a=\ui u + \tfrac12$. We therefore record the forest
expansion at the shifted argument. Writing $\mathcal T_\ell$ for the set
of trees generated at order $\ell$, and $m$ and $x$ for the numbers of
leaves \M{} and \X{} of a tree $\tau$,
\beq
\tmF_\ell(u-\ui/2)=\sum_{\tau\in\mathcal T_\ell}
w_\tau\,\lambda_a^{\,m}\,a^{\,x}\;\tau ,
\label{eq:forestweights}
\eeq
which defines the weight $w_\tau>0$ of each tree. The $w_\tau$ are
symmetry factors: in \eqref{eq:tmFrecursion} an unordered pair of
distinct subtrees is produced twice, cancelling the factor $\tfrac12$,
whereas a pair of identical subtrees is produced once. Hence
$w_\tau=2^{-s(\tau)}$, where $s(\tau)$ is the number of nodes of $\tau$
whose two subtrees coincide; for example $w_{\MXdXd}=1$ and
$w_{\MMd}=\tfrac12$. For more on the combinatorics of these trees, see
\cite[Section 1.2 and Remark 1.1]{friz2022forests}.

We write
\beq
r_\tau(a):=\lambda_a^{\,m-1}\,a^{x},
\label{eq:rtau}
\eeq
so that the contribution of $\tau$ to $\psi(T;u-\ui/2)$ is
$w_\tau\,\lambda_a\,r_\tau(a)\,\tau$; thus, up to its symmetry factor,
$\tau$ enters the CGF only through its leaf content $(m,x)$. One factor
of $\lambda_a$ is held back
because, at $a=\tfrac12+\ui u$, it cancels the denominator of the
Fourier measure $du/(u^2+\tfrac14)$ in \eqref{eq:formalEqn} below, up to
the constant $-\tfrac12$.

The \emph{order} of a tree $\tau\in\mathcal T_\ell$ is $\ell$; orders add
under products, independently of the number of tree factors. A tree of
order $\ell$ with $m$ leaves $\M$ and $x$ leaves $\X$ satisfies
\begin{equation}
\ell=2m+x-2.
\label{eq:treegrading}
\end{equation}

For fixed maturity $T$, let $\Sigma(k):=T\,\sigma_{\mathrm{BS}}^2(k,T)$
denote the total implied
variance, and let $k=\log(K/S_0)$ denote the log-strike.
Regarding \eqref{eq:CGFmF} as
a formal power series in $\epsilon$ whose power counts the forest index
$\ell$, we write the Bergomi--Guyon expansion of the total implied
variance smile as
\beq
\Sigma(k) =  \sum_{\ell=0}^\infty\,\epsilon^\ell\, a_\ell(k).
\label{eq:BG}
\eeq
Its coefficients are determined by the matching condition of
\cite[Equation (5.7)]{gatheral2006volatility}. After setting
$a_0(k)=M$, this condition becomes
\begin{multline}
\int_0^\infty\frac{du}{u^2+\frac{1}{4}}
\,e^{- \frac{1}{2}\bigl(u^2+\tfrac14\bigr)\,M}\,\Re \Bigl[e^{-\ui uk}\,
\,\exp\left\{
 \sum_{\ell=1}^\infty\,\epsilon^\ell\,\tilde{\mF}_\ell(u-\ui/2)
\right\}
\Bigr] \\
=
\int_0^\infty\frac{du}{u^2+\frac{1}{4}}\,e^{- \frac{1}{2}\bigl(u^2+\tfrac14\bigr)\,M}\,
\,\Re\Bigl[e^{-\ui uk}\,\exp\Bigl\{-\frac{1}{2}\bigl(u^2+\tfrac{1}{4}\bigr) \,
\sum_{\ell \geq 1}\,\epsilon^\ell \, a_\ell(k)
\Bigr\}\Bigr],
\label{eq:formalEqn}
\end{multline}
matched order by order in $\epsilon$.
\subsection{The form of the BG expansion}
Following \cite[eq.~(A2)]{alos2020exponentiation}, let
$I_0$ denote the Gaussian kernel
\[
I_0(k)=\int_0^\infty du\,
\Re\Bigl[e^{-\ui uk}\,e^{-\frac12(u^2+\frac14)M}\Bigr]
=\sqrt{\frac\pi2}\;\frac{e^{-M/8}}{\sqrt M}\;e^{-k^2/2M}.
\]
Also, differentiating $I_0$ $j$ times gives
\[(-\p_k)^{j}\,I_0(k)=
\int_0^\infty du\,
\Re\left[e^{-\ui uk}\,e^{-\frac12(u^2+\frac14)M}\,(\ui u)^{j} \right]
=:I_0(k) \, \tilde I_{j}(k),
\]
where
\[
\tilde I_j(k):=\frac{(-\p_k)^{j}\,I_0(k)}{I_0(k)}
=\frac{1}{M^{j/2}}\,\mathrm{He}_{j}\Bigl(\frac k{\sqrt M}\Bigr),
\]
and $\mathrm{He}_j$ denotes the probabilists' Hermite polynomial of
degree $j$ \cite[Chapter 22]{abramowitz1964handbook}. 
The coefficients $a_\ell(k)$ in \eqref{eq:BG} are linear combinations of trees and products of trees; we call the polynomial multiplying a tree or product of trees $\btau$ in $a_\ell$ its \emph{prefactor}, and denote it $c_{\btau}$. The prefactors are linear combinations of the $\tilde I_j$. For example,  at
second order, the BG expansion reads
\begin{equation}\label{eq:BG2}
\begin{aligned}
\Sigma(k)
&=
\M +\epsilon  \left(\frac{k}{M}+\frac{1}{2}\right)\, \MXd
+
\frac{1}{4} \epsilon^2 \left(\frac{k^2}{M^2}-\frac{1}{M}-\frac{1}{4}\right)\,\MMd
\\
&+\epsilon^2 \left(\frac{k^2}{M^2}+\frac{k}{M}-\frac{1}{M}+\frac{1}{4}\right)\,\MXdXd
+\frac{ \epsilon^2 }{4 M}\,\left(-\frac{5 k^2}{M^2}-\frac{2 k}{M}+\frac{3}{M}+\frac{1}{4}\right)\,(\MXd)^2
+\cO(\epsilon^3).
\end{aligned}
\end{equation}
%
\begin{lemma}\label{lem:singleTree}
For a single tree $\tau$, expand
$r_\tau(a)=\sum_j\varrho_j\,(a-\tfrac12)^j$.  Then
\begin{equation}
c_\tau(k)=w_\tau\sum_j\varrho_j\,\tilde I_j(k).
\label{eq:c_tau}
\end{equation}
Equivalently, $c_\tau$ is obtained from $w_\tau r_\tau$ by expanding in
powers of $a-\tfrac12$ and replacing each power $(a-\tfrac12)^j$ by
$\tilde I_j$.
\end{lemma}

\begin{proof}
Matching the terms of \eqref{eq:formalEqn} linear in trees, the factor
$\lambda_a$ common to both sides cancels against the measure
$du/(u^2+\tfrac14)$, since $\lambda_a=-\tfrac12(u^2+\tfrac14)$ at
$a=\tfrac12+\ui u$, leaving the polynomial
$w_\tau r_\tau(\tfrac12+\ui u)=w_\tau\sum_j\varrho_j\,(\ui u)^j$. Since
$(\ui u)^j e^{-\ui uk}=(-\p_k)^j e^{-\ui uk}$, each power $(\ui u)^j$
integrates to $\tilde I_j\,I_0$, as in
\cite[Appendix A.2.1]{alos2020exponentiation}, and dividing by $I_0$ gives
\eqref{eq:c_tau}.
\end{proof}
\subsection{The Bergomi--Guyon operator}\label{sec:operator}
The substitution of Lemma \ref{lem:singleTree} is generated by a single
operator. For a polynomial $q\in\mathbb{R}[a]$, define
\begin{equation}
\cB[q](k):=\frac{q\bigl(\tfrac12-\p_k\bigr)\,I_0(k)}{I_0(k)}.
\label{eq:BDef}
\end{equation}
In particular,
\[
\cB\bigl[(a-\tfrac12)^j\bigr](k)
=\frac{(-\p_k)^j I_0(k)}{I_0(k)}
=\tilde I_j(k).
\]
Consequently, Lemma \ref{lem:singleTree} gives
$c_\tau(k)=w_\tau\,\cB[r_\tau](k)$ for every tree $\tau$: up to the
symmetry factor $w_\tau$, the \emph{Bergomi--Guyon operator} (or
\emph{BG operator}) $\cB$ maps the polynomial multiplying a tree in the
cumulant generating function \eqref{eq:CGFmF} to its prefactor in the
smile expansion \eqref{eq:BG}.

\begin{lemma}\label{lem:Hermite}
In the variables $\zeta:=\frac12+\frac k M$ and $\theta:=\frac{1}{M}$,
the BG operator can be written as
\begin{equation}
\cB[q](\zeta,\theta)=\bigl(e^{-\frac\theta2\,\p_{\zeta}^{2}}\,q\bigr)(\zeta) = \sum_{j\geq0}\frac{(-\theta/2)^j}{j!}\,q^{(2j)}(\zeta).
\label{eq:Bheat}
\end{equation}
\end{lemma}
\begin{proof}
For any function $f$,
\[
\bigl(I_0^{-1}\,\p_k\,I_0\bigr)f
=\frac{\p_k(I_0\,f)}{I_0}
=\p_k f+\frac{\p_k I_0}{I_0}\,f
=\Bigl(\p_k-\frac kM\Bigr)f ,
\]
since $\p_k I_0/I_0=-k/M$.  Conjugation $A\mapsto I_0^{-1}A\,I_0$ preserves
sums and products, hence
$I_0^{-1}\,q(\p_k)\,I_0=q(\p_k-k/M)$ for every polynomial
$q$ and, applying both sides to the constant function $1$,
\[
q(\p_k)\,I_0/I_0=q(\p_k-k/M)\,1.
\]
Substituting into \eqref{eq:BDef} gives
\[
\cB[q](k)= q\left(\tfrac12 -\p_k+\frac k M\right)1 = q(\zeta-\theta\,\p_\zeta)\,1.
\]
Hadamard's lemma expands a conjugation in iterated
commutators:
\beq
e^{A}\,B\,e^{-A}=B+[A,B]+\tfrac1{2!}\,\bigl[A,[A,B]\bigr]+\cdots.
\label{eq:Hadamard}
\eeq
Take $A=-\tfrac\theta2\p_{\zeta}^{2}$ and $B=\zeta$.  The commutator $\bigl[-\tfrac\theta2\p_{\zeta}^{2},\zeta\bigr]
=-\theta\,\p_\zeta$ commutes with $\p_\zeta$, so \eqref{eq:Hadamard} stops after one term and becomes
\[
e^{-\frac\theta2\p_{\zeta}^{2}}\;\zeta\;e^{+\frac\theta2\p_{\zeta}^{2}}
=\zeta-\theta\,\p_\zeta,
\]
and the result follows:
$q(\zeta-\theta\,\p_\zeta)\,1
=e^{-\frac\theta2\p_{\zeta}^{2}}\,q(\zeta)\,e^{+\frac\theta2\p_{\zeta}^{2}}\,1
=e^{-\frac\theta2\p_{\zeta}^{2}}\,q(\zeta)$.
\end{proof}
Up to a constant, $I_0$ is a Gaussian density in $k$ with variance $M$,
hence a Gaussian density in $\zeta=\tfrac12+k/M$ with variance
$1/M=\theta$. Accordingly,
$\cB$ maps shifted monomials to scaled probabilists' Hermite
polynomials:
\[
\cB\bigl[(a-\tfrac12)^j\bigr](\zeta,\theta)
=\theta^{j/2}\,\He_j\Bigl(\frac{\zeta-\tfrac12}{\sqrt\theta}\Bigr).
\]
This expression is a polynomial in $(\zeta,\theta)$ and therefore extends
to $\theta=0$. At $\zeta=\tfrac12+k/M$ and $\theta=1/M$, it equals
$\tilde I_j(k)$.
\subsubsection{Examples}
\label{sec:single2}
In terms of the operator $\cB$, it is straightforward to compute the prefactors $c_\tau=w_\tau\,\cB[r_\tau]$ of single trees:
\begin{align*}
\MXd:& \qquad \cB[a] = \zeta = \frac12+\frac k M,\\
\MXdXd:& \qquad \cB[a^2] = \zeta^2 - \theta = \left(\frac12+\frac k M\right)^2- \frac 1 M = \frac{k^2}{M^2}+ \frac k M -\frac 1 M + \frac 14,\\
\MMd:&\qquad \tfrac12\,\cB[\tfrac12 a\,(a-1)] = \frac14 \,\left(\zeta^2 - \theta  - \zeta \right)=\frac14  \left(  \frac{k^2}{M^2} -\frac 1 M - \frac 14  \right),
\end{align*}
where $w_{\MXd}=w_{\MXdXd}=1$ and $w_{\MMd}=1/2$.  These prefactors
are all in agreement with \eqref{eq:BG2}.  Note in particular how simple
they look in terms of $\zeta$ and $\theta$.
\subsubsection{Choice of variables}
From now on we work in the variables
$(\zeta,\theta)$ of Lemma \ref{lem:Hermite}. At each order in $\epsilon$,
the coefficients considered below are polynomials in $\zeta$ and
$\theta$ whose coefficients are trees and products of trees. We regard $\theta$ as a formal
parameter, which allows evaluation at $\theta=0$. We sometimes display
the arguments explicitly, writing, for example,
$\Sigma(\zeta,\theta)$, and sometimes write simply $\Sigma$.
\subsection{The matching condition in operator form}
\begin{prop}\label{prop:matching}
Let
\begin{equation}
\tilde \psi(a):=\sum_{\ell\ge1}\epsilon^\ell\sum_{\tau\in\mathcal{T}_\ell}w_\tau\,\lambda_a\,r_\tau(a)\,\tau,
\qquad
\tilde \Sigma(\zeta,\theta):=\sum_{\ell\ge1}\epsilon^\ell\,a_\ell(\zeta,\theta),
\label{eq:gdef}
\end{equation}
so that $\psi(T;u-\ui/2)=\lambda_a\,M+\tilde \psi(a)$ at
$a=\tfrac12+\ui u$, and $\tilde \Sigma(\zeta,\theta) = \Sigma(\zeta,\theta)-M$.  Since every
tree carries at least one leaf $\M$, hence at least one factor
$\lambda_a$, the quotient $\tilde \psi/\lambda_a$ is again a polynomial
in $a$.  The matching condition \eqref{eq:formalEqn} can be rewritten as
\begin{equation}
\cB\Bigl[\frac{e^{\tilde \psi(a)}-1}{\lambda_a}\Bigr]
\;=\;\sum_{n\ge1}\frac{{\tilde \Sigma(\zeta,\theta)}^n}{n!}\;\cB\bigl[\lambda_a^{\,n-1}\bigr] .
\label{eq:matchingexp}
\end{equation}
\end{prop}
\begin{proof}
Proceed as in the proof of Lemma \ref{lem:singleTree} but keep the
nonlinear terms: subtract the zeroth-order terms, expand the exponentials,
and rewrite the Fourier measure using
$1/(u^2+\tfrac14)=-1/(2\lambda_a)$ at $a=\tfrac12+\ui u$.
Cancel the common factor $-\tfrac12$,
substitute $u^ne^{-\ui uk}=(\ui\p_k)^ne^{-\ui uk}$, and divide by
$I_0$.  The powers of $\tilde\Sigma$, which depend on $k$ but not on $u$,
pass outside the $u$-integral.
\end{proof}

\begin{remark}\label{rem:thetazero}
At $\theta=0$, the operator $\cB$ in \eqref{eq:Bheat} maps any polynomial $q(a)$ to $q(\zeta)$, which amounts to the substitution
$a\mapsto\zeta$. Then \eqref{eq:matchingexp} becomes
\[
\frac {e^{\tilde \psi(\zeta)}-1}{\lambda_\zeta}=
\sum_{n\ge1}\frac{{\tilde \Sigma(\zeta,0)}^n}{n!}\;\lambda_\zeta^{\,n-1}
=
\frac {e^{\lambda_\zeta\,\tilde \Sigma(\zeta,0)}-1}{\lambda_\zeta}.
\]
Both exponents have zero constant term in $\epsilon$; applying the formal
logarithm gives
\beq
\tilde \Sigma(\zeta,0) = \frac{\tilde \psi(\zeta)}{\lambda_\zeta},
\label{eq:thetazero}
\eeq
which will serve as the initial condition for the recursion of Theorem
\ref{thm:recursion} below.  Since $\tilde \psi$ is linear in the trees,
so is $\tilde \Sigma(\zeta,0)$. At $\theta=0$, the prefactor of every
product of two or more trees vanishes.
\end{remark}
\section{A recursion algorithm}\label{sec:recursion}
Equation \eqref{eq:matchingexp} already determines the prefactors of
trees and products of trees recursively, as in the Bell polynomial
algorithm of \cite{alos2020exponentiation}.  The LHS of
\eqref{eq:matchingexp} involves the natural generalization of $r_\tau(a)$
to products of trees.  At order $\epsilon^\ell$, the moments
$\cB[\lambda_a^{n-1}]$ on the RHS are found to cancel terms of order
higher than $k^\ell$, as in the second-order example of Section
\ref{sec:example}.  In this section, we derive a more efficient
algorithm that exhibits none of these higher-order cancellations.
\subsection{Moments and cumulants}
We start by noticing that the RHS of \eqref{eq:matchingexp} looks like an
integral of a moment generating function.  We are then led to define the
following (formal) moment generating function and the corresponding
cumulant generating function:
\[
G(y;\zeta,\theta):=\cB\bigl[e^{y\lambda_a}\bigr]
=\sum_{s\ge0}\frac{y^s}{s!}\,\cB[\lambda_a^s],
\qquad
K(y;\zeta,\theta):=\log G(y;\zeta,\theta)
=\sum_{m\ge1}K_m(\zeta,\theta)\,y^m .
\]
We refer to the $\cB[\lambda_a^s]$ as formal moments and to the $K_m$ as
normalized cumulants. Note that no underlying positive probability law is implied and the normalization of $K$ has no factor $m!$.
Define
\[
\cG(y;\zeta,\theta):=\sum_{r\ge1}\frac{y^r}{r!}\;\cB\bigl[\lambda_a^{\,r-1}\bigr],
\]
so that $G = \p_y\cG$.  In terms of $\cG$, the matching condition
\eqref{eq:matchingexp} becomes
\begin{equation}
\cB\Bigl[\frac{e^{\tilde \psi(a)}-1}{\lambda_a}\Bigr]
=\cG\bigl(\tilde \Sigma(\zeta,\theta);\zeta,\theta\bigr)
=\int_0^{\tilde \Sigma(\zeta,\theta)} e^{K(y;\zeta,\theta)}\,dy .
\label{eq:matchingG}
\end{equation}
The following proposition gives a closed-form expression for the CGF $K(y;\zeta,\theta)$.
\begin{prop}\label{prop:K}
With $\kappa:=\zeta-\tfrac12=\frac k M$,
\begin{equation}
K(y;\zeta,\theta) =\frac{y\,\kappa^2}{2\,(1+\theta \,y)}-\frac12\,\log(1+\theta\, y)-\frac y8 ,
\label{eq:Kclosed}
\end{equation}
so that $K_1(\zeta,\theta)=\lambda_\zeta-\tfrac\theta2$ and
\begin{equation}\label{eq:Km}
K_m(\zeta,\theta)=(-\theta)^{m-1}\Bigl(\frac{\kappa^2}{2}-\frac\theta{2 m}\Bigr),
\qquad m\ge2 .
\end{equation}
\end{prop}
\begin{proof}
Since $\lambda_\zeta=\tfrac12\,\left(\kappa^2-\tfrac14\right)$ and
$\p_\kappa =\p_\zeta$, we have
$
e^{y\,\lambda_\zeta} = e^{-\frac y 8}\,e^{\frac12 y\,\kappa^2}
$
and,
by \eqref{eq:Bheat} and the definition of $G$,
\[
G(y)=\cB\left[e^{y\,\lambda_a} \right]=
e^{-\frac\theta2\p_{\zeta}^{2}}\,e^{y\,\lambda_\zeta}=
e^{-y/8}\,e^{-\frac\theta2\p_{\kappa}^{2}}e^{\frac12\,y\,\kappa^2}.
\]
Now consider the following identity of formal power series in $\alpha$:
\begin{equation}
e^{-\frac\theta 2\,\p_{\kappa}^{2}}\,e^{\frac12\,\alpha\, \kappa^2}
=\frac 1{(1+\theta\,\alpha)^{1/2}}\,
\exp\left\{\frac{\alpha\,\kappa^2}{2\,(1+\theta\,\alpha)}\right\}.
\label{eq:mehler}
\end{equation}
Both sides satisfy
$\p_\theta F=-\tfrac12\,\p_{\kappa}^{2}F$ and have the same value at
$\theta=0$; hence they agree as formal power series in $\alpha$.
Setting $\alpha=y$ and taking the logarithm gives
\eqref{eq:Kclosed}.
Finally, expanding $(1+\theta y)^{-1}$ and $\log(1+\theta y)$
in $y$ gives the $K_m$.
\end{proof}
\noindent
\begin{cor}\label{cor:cumulants}
$K(y;\zeta,0)=y\,\lambda_\zeta$, and $K_m(\zeta,\theta)=\cO(\theta^{m-1})$
for every $m\ge1$.
\end{cor}
\begin{proof}
Immediate from \eqref{eq:Kclosed}: at $\theta=0$ the right-hand side is
$\tfrac12\,y\,\kappa^2-\tfrac y8=y\,\lambda_\zeta$, and the expansion of
Proposition \ref{prop:K} shows that $K_m$ carries the overall factor
$\theta^{m-1}$.
\end{proof}
\subsection{The recursion}
By \eqref{eq:Bheat}, $\cB[q]$ satisfies the heat equation: for every
polynomial $q$, differentiating
$e^{-\frac\theta2\,\p_{\zeta}^{2}}\,q(\zeta)$ in $\theta$ brings down
$-\tfrac12\,\p_{\zeta}^{2}$, so
\beq
\bigl(\p_\theta+\tfrac12\,\p_{\zeta}^{2}\bigr)\,\cB[q]=0 .
\label{eq:Bharmonic}
\eeq
Both sides of the matching condition \eqref{eq:matchingG} are built,
order by order, from such terms.  This forces the relation between
$\tilde\Sigma$ and the normalized cumulants of Proposition \ref{prop:K} given in
the following theorem.
\begin{theorem}\label{thm:recursion}
In the variables $(\zeta,\theta)$ of Lemma \ref{lem:Hermite},
\begin{equation}
\bigl(\p_\theta+\tfrac12\p_{\zeta}^{2}\bigr)\,\tilde \Sigma
=-\,\p_\zeta \tilde \Sigma\,\left[\tfrac12\,\p_y K\,\p_\zeta \tilde \Sigma+\p_\zeta K\,\right]_{y=\tilde \Sigma(\zeta,\theta)} ,
\qquad
\tilde \Sigma(\zeta,0)=\frac{\tilde \psi(\zeta)}{\lambda_\zeta },
\label{eq:recursion}
\end{equation}
with $K(y;\zeta,\theta)$ as in Proposition \ref{prop:K}; in particular,
$K$ is independent of the model and of the trees.
\end{theorem}
\begin{proof}
Write the matching condition \eqref{eq:matchingG} as
$\cL = \cG(\tilde \Sigma)$, where $\cL$ denotes its left-hand side.
By \eqref{eq:Bharmonic}, $\cL$ and, at each fixed $y$, $\cG(y)$
satisfy the heat equation, each being, order by order, $\cB$ of a
polynomial in $a$. Now $\cG(\tilde\Sigma)$ depends on
$(\zeta,\theta)$ both through the explicit dependence of $\cG$ and
through the argument $y=\tilde \Sigma(\zeta,\theta)$, so, by the chain
rule,
\[
\p_\theta\cL=\p_y\cG\;\p_\theta\tilde \Sigma+\p_\theta\cG,
\qquad
\p_{\zeta}^{2} \cL
=\p_y\cG\;\p_{\zeta}^{2}\tilde \Sigma
+\p_y^2\cG\,\bigl(\p_\zeta\tilde \Sigma\bigr)^2
+2\,\p_y\p_\zeta\cG\;\p_\zeta\tilde \Sigma
+\p_{\zeta}^{2}\cG ,
\]
with every derivative of $\cG$ evaluated at $y=\tilde \Sigma$.
Combining,
\[
0=\left( \p_\theta+\tfrac12\p_{\zeta}^{2} \right)\cL
=\p_y \cG\,\left(\p_\theta \tilde \Sigma +\tfrac12\p_{\zeta}^{2} {\tilde \Sigma}\right)
+\tfrac12\,\p_y^2\cG\,\bigl(\p_\zeta\tilde \Sigma\bigr)^2
+\p_y\p_\zeta \cG\;\p_\zeta\tilde \Sigma
+\left(\p_\theta \cG+\tfrac12\p_{\zeta}^{2}\cG\right),
\]
where the last bracket vanishes. Since $\p_y\cG=G=e^K$ has constant term
$1$, it is invertible as a formal power series. Dividing by $\p_y\cG$
and using $\p_y^2\cG/\p_y\cG=\p_yK$ and
$\p_y\p_\zeta\cG/\p_y\cG=\p_\zeta K$ gives \eqref{eq:recursion}.  The initial condition is \eqref{eq:thetazero}.
\end{proof}

\subsection{Application to the BG smile expansion}

Write
\[
\tilde \Sigma=\sum_{j\ge1}\tilde \Sigma_j ,
\]
where $\tilde \Sigma_j$ collects the products $\btau=\tau_1\cdots\tau_j$
of exactly $j$ trees, $\tau_i\in\mathcal{T}_{\ell_i}$, of total order
$\ell:=\ell_1+\cdots+\ell_j$, each multiplied by its prefactor
$c_{\btau}$; in particular
$\tilde \Sigma_1=\cB\bigl[\tilde\psi/\lambda_a\bigr]$.

\begin{remark}\label{rem:leafcontent}
The symmetry factors factor out of the whole computation.  Let
$\bar c_{\btau}$ denote the prefactor obtained with every $w_\tau$
replaced by $1$, which by \eqref{eq:rtau} depends on the $\tau_i$ only
through their leaf contents $(m_i,x_i)$, and put
$w_{\btau}:=w_{\tau_1}\cdots w_{\tau_j}$.  Then
\beq
c_{\btau}=w_{\btau}\,\bar c_{\btau}.
\label{eq:weightsout}
\eeq
Indeed, \eqref{eq:matchingexp} is an identity between polynomials in the
tree symbols with tree-free coefficients, and determines $\tilde\Sigma$
from $\tilde\psi$.  The substitution $\tau\mapsto w_\tau\,\tau$ respects
sums and products and commutes with $\cB$, and carries $\tilde\psi$
with unit weights to $\tilde\psi$; applied to both sides of
\eqref{eq:matchingexp}, it therefore carries the solution
$\tilde\Sigma$ with unit weights to the solution with the weights
$w_\tau$, multiplying the prefactor of $\btau$ by $w_{\btau}$.
In an implementation, a forest is thus specified by the leaf contents
and weights of its trees, and the weights enter only at the end.
\end{remark}
The following corollary of Theorem \ref{thm:recursion} shows how to
compute the prefactors of products of trees in the BG smile expansion.
\begin{cor}\label{cor:determined}
For each $j\ge2$, the $j$-tree part $\tilde \Sigma_j$ of the smile
satisfies a heat equation with zero initial condition,
\begin{equation}
\bigl(\p_\theta+\tfrac12\,\p_{\zeta}^{2}\bigr)\,\tilde \Sigma_j=S_j,
\qquad
\tilde \Sigma_j(\zeta,0)=0 ,
\label{eq:layer}
\end{equation}
whose source $S_j$ is the
 $j$-tree part of
 \beq
 \label{eq:recursionRHS}
 -\,\p_\zeta \tilde \Sigma\,\left[\tfrac12\,\p_y K\,\p_\zeta \tilde \Sigma+\p_\zeta K\,\right]_{y=\tilde \Sigma(\zeta,\theta)}
,
\eeq
the right-hand side of \eqref{eq:recursion}, and depends only on
$\tilde \Sigma_i$ with $i < j$ and the normalized cumulants $K_m$ of
Proposition \ref{prop:K}.  Since \eqref{eq:layer} has exactly one
polynomial solution at each order in $\epsilon$, the recursion determines,
by induction on $j$, the whole of
$\tilde \Sigma$ and with it every prefactor $c_{\btau}$ of the BG
expansion.
\end{cor}
\begin{proof}
Sort both sides of \eqref{eq:recursion} by the number of trees.  The
left-hand side is linear in $\tilde \Sigma$, so its $j$-tree part is the
left-hand side of \eqref{eq:layer}.  

By definition of the normalized cumulants, $K = \sum_{m\ge1}K_m\,y^m$.  Setting
$y=\tilde \Sigma(\zeta,\theta)$, \eqref{eq:recursionRHS} then becomes
\begin{equation}
-\p_\zeta\tilde \Sigma\,\left[\frac12 \sum_{m\ge 1}
m\,K_m\,
\tilde \Sigma^{m-1}\,\p_\zeta\tilde \Sigma
\;+\;\sum_{m\ge 1}\p_\zeta K_m\,\tilde \Sigma^m
\right].
\label{eq:sourceexpanded}
\end{equation}
Writing
$\tilde \Sigma=\sum_i\tilde \Sigma_i$,  \eqref{eq:sourceexpanded} is a sum of products of
at least two trees.  In the $j$-tree part, each factor therefore
carries fewer than $j$ trees: $S_j$ is built from the
$\tilde \Sigma_i$ with $i<j$ only, and $\tilde \Sigma_j$ satisfies
\eqref{eq:layer};  it follows from Remark \ref{rem:thetazero} that the prefactor of every product
of two or more trees vanishes at $\theta=0$.
To see that \eqref{eq:layer} has exactly one polynomial solution at each
order in $\epsilon$, note that, viewed as a polynomial in $\theta$,
$\tilde \Sigma_j(\zeta,\theta)$ is determined by its derivatives at
$\theta=0$.
These may be computed from
\eqref{eq:layer},
starting from $\tilde \Sigma_j(\zeta,0)=0$, as follows:
\begin{equation}
\tilde \Sigma^{(n)}_j:=\p_\theta^{\,n}\,\tilde \Sigma_j\bigr|_{\theta=0}
=\p_\theta^{\,n-1}S_j\bigr|_{\theta=0}
-\tfrac12\,\p_{\zeta}^{2}\,\tilde \Sigma^{(n-1)}_j,
\qquad n\ge1 ,
\label{eq:integrate}
\end{equation}
with $\tilde \Sigma^{(0)}_j=\tilde \Sigma_j(\zeta,0)=0$.
For each $n \geq 1$, $\tilde \Sigma^{(n)}_j$ is a sum of products of polynomials in $\zeta$ and trees.  Taylor expansion then gives the explicit expression
\beq
\label{eq:Sigmaj}
\tilde \Sigma_j(\zeta,\theta) = \sum_{n=1}^\infty\,\frac{\tilde \Sigma^{(n)}_j\,\theta^n}{n!}.
\eeq
The induction on $j$ starts from
$\tilde \Sigma_1=\cB\bigl[\tilde \psi/\lambda_a\bigr]$, which solves
\eqref{eq:recursion} for single trees.
\end{proof}
\begin{remark}
At each order in $\epsilon$, the sum in \eqref{eq:Sigmaj} is finite:
once $n-1$ exceeds the degree of $S_j$ in $\theta$,
\eqref{eq:integrate} gives
$\tilde \Sigma^{(n)}_j=-\tfrac12\,\p_{\zeta}^{2}\tilde \Sigma^{(n-1)}_j$,
which lowers the degree in $\zeta$ by two at each step and therefore
vanishes after finitely many steps.
\end{remark}

\begin{remark}\label{rem:timings}
We have implemented the algorithm \eqref{eq:integrate} in Python 3.12.  Its
output coincides with that of the moment computation of
\cite[Appendix A.1]{alos2020exponentiation} as far as we have compared
them, through order six.  The implementation is fast: on a MacBook Air
(M5, 24 GB RAM), it generates the expansion through order six in $0.086$\,s and
through order ten in $19.3$\,s (median of 5 runs, including imports and rendering).  Appendix~\ref{app:BG4} lists the
expansion through order four.
\end{remark}

\subsection{The BG expansion to  $\cO(\epsilon^2)$}\label{sec:example}
The second-order expansion \eqref{eq:BG2} was obtained in
\cite[Appendix A.1]{alos2020exponentiation} by matching Bell-polynomial
terms. Here is the same calculation using Corollary
\ref{cor:determined}.
The single tree prefactors are already given by Lemma \ref{lem:singleTree} and exhibited in Section \ref{sec:single2}, so we already have $\tilde \Sigma_1(\zeta,\theta)$.
At order $\epsilon^2$, the only term in $\tilde \Sigma_2(\zeta,\theta)$ is the
product $(\MXd)^2$.
The moment matching method of \cite{alos2020exponentiation} is equivalent
to solving for the prefactor of $(\MXd)^2$
using
\eqref{eq:matchingexp}, as follows:
Keeping only the terms relevant to the coefficient of
$\epsilon^2(\MXd)^2$, the CGF and the smile have the form
\[
\tilde\psi(a)=\epsilon\,\lambda_a a\,\MXd+\cdots,
\qquad
\tilde\Sigma=\epsilon\,c_{\MXd}\,\MXd
+\epsilon^2c_{(\MXd)^2}(\MXd)^2+\cdots,
\]
where the omitted terms do not contribute to $\epsilon^2(\MXd)^2$ and
$c_{\MXd}=\cB[a]$. On the left-hand side of
\eqref{eq:matchingexp}, the coefficient of $\epsilon^2(\MXd)^2$ comes from
$\tfrac12\tilde\psi^2$ and is
\[
\cB\left[\frac{1}{\lambda_a}\,
\frac12(\lambda_a a)^2\right]
=\frac12\cB[\lambda_a a^2].
\]
On the right-hand side, the term with $n=1$ contributes the unknown
$c_{(\MXd)^2}$, while the term with $n=2$ contributes
$\tfrac12c_{\MXd}^2\cB[\lambda_a]$. Hence
\[
c_{(\MXd)^2}
=\frac12\cB[\lambda_a a^2]
-\frac12\cB[a]^2\cB[\lambda_a].
\]
By \eqref{eq:Bheat},
\[
\tfrac12\,\cB\bigl[\lambda_a a^2\bigr]
=\tfrac14\bigl(\zeta^4-\zeta^3-6\theta \zeta^2+3\theta \zeta+3\theta^2\bigr),
\qquad
\tfrac12\,\cB[a]^2\,\cB[\lambda_a]
=\tfrac14\bigl(\zeta^4-\zeta^3-\theta \zeta^2\bigr).
\]
The quartic
and cubic terms cancel identically, giving
\begin{equation}
c_{(\MXd)^2}
=-\frac\theta4\,\bigl(5\zeta^2-3\zeta-3\theta\bigr)
=\frac1{4M}\left(-\frac{5k^2}{M^2}-\frac{2k}{M}+\frac3M+\frac14\right),
\label{eq:mystery}
\end{equation}
as in \eqref{eq:BG2}. This is
the mysterious cancellation referred to in \cite[Remark A.1]{alos2020exponentiation}.
On the other hand, we can apply Corollary \ref{cor:determined} with
$j=2$.  For $\MXd$, we have $r_{\MXd}(a)=a$, so $c_{\MXd}=\zeta$ and
$\p_\zeta c_{\MXd}=1$.  Also, from Proposition \ref{prop:K},
$K_1=\lambda_\zeta-\tfrac\theta2=\tfrac12\,\bigl(\kappa^2-\tfrac14-\theta\bigr)$
and $\p_\zeta K_1=\kappa$, so the $(\MXd)^2$-part of
\eqref{eq:sourceexpanded} is
\begin{align*}
S_2&=-\tfrac12\,K_1\,\bigl(\p_\zeta c_{\MXd}\bigr)^2
-\p_\zeta K_1\;c_{\MXd}\,\p_\zeta c_{\MXd}
=-\Bigl[\tfrac12\,K_1+\zeta\,\kappa\Bigr]
=-\tfrac54\,\zeta^2+\tfrac34\,\zeta+\tfrac\theta4.
\end{align*}
No quartic or cubic terms appear.
Using \eqref{eq:integrate}, we may compute derivatives of the prefactor at $\theta=0$ as follows:
\[
\p_\theta\,c_{(\MXd)^2}(\zeta,0)=-\tfrac54\,\zeta^2+\tfrac34\,\zeta ,
\qquad
\p_\theta^{\,2}c_{(\MXd)^2}(\zeta,0)
=\tfrac14-\tfrac12\,\p_{\zeta}^{2}\bigl(-\tfrac54\,\zeta^2+\tfrac34\,\zeta\bigr)
=\tfrac32.
\]
Taylor expansion in $\theta$ then gives
\begin{equation}
c_{(\MXd)^2}
=-\frac\theta4\,\bigl(5\zeta^2-3\zeta-3\theta\bigr)
=\frac1{4M}\left(-\frac{5k^2}{M^2}-\frac{2k}{M}+\frac3M+\frac14\right),
\label{eq:mystery2}
\end{equation}
in agreement with the moment matching method.
\subsubsection{The origin of the mysterious cancellation}
Expanding \eqref{eq:Bheat},
\[
\cB[q]=q(\zeta)-\tfrac\theta2\,\p_{\zeta}^{2}q(\zeta)
+ \cdots :
\]
each successive correction raises the exponent of $\theta$ by one and
applies two additional derivatives with respect to $\zeta$. 
Thus
\[
\tfrac12\,\cB[\lambda_a\,a^2]
=\tfrac12\,\lambda_\zeta\,\zeta^2+\theta\,q_1(\zeta,\theta),
\qquad
\tfrac12\,\cB[a]^2\,\cB[\lambda_a]
=\tfrac12\,\zeta^2\,\lambda_\zeta+\theta\,q_2(\zeta,\theta),
\]
where $q_1$ and $q_2$ are polynomials of degree at most two in $\zeta$.
The two leading polynomials are identical, so their difference
$c_{(\MXd)^2}=\theta\,(q_1-q_2)$ is $\theta$ times a polynomial of
degree at most two in $\zeta$.  This is the origin of the cancellation in
\cite[Remark A.1]{alos2020exponentiation}.

On the other hand, in the recursion algorithm, the cancellation of leading terms is implicit in the construction of the source term $S_j$.
By \eqref{eq:Bharmonic},
every moment satisfies the heat equation.  Apply the operator $\p_\theta+\tfrac12\,\p_{\zeta}^{2}$ to the two
moment expressions.  The quartic moment
$\tfrac12\,\cB[\lambda_a\,a^2]$ is annihilated outright.  In the
product $\tfrac12\,\cB[a]^2\,\cB[\lambda_a]$, each factor is
annihilated, and only the following cross terms survive:
\[
\bigl(\p_\theta+\tfrac12\,\p_{\zeta}^{2}\bigr)
\left\{\tfrac12\,\cB[a]^2\,\cB[\lambda_a]\right\}
=\tfrac12\,\bigl(\p_\zeta\cB[a]\bigr)^2\,\cB[\lambda_a]
+\cB[a]\;\p_\zeta\cB[a]\;\p_\zeta\cB[\lambda_a]
=\tfrac12\,K_1+\zeta\,\kappa ,
\]
since $\cB[a]=\zeta$ and $\cB[\lambda_a]=K_1$.  Subtracting,
\[
\bigl(\p_\theta+\tfrac12\,\p_{\zeta}^{2}\bigr)\,c_{(\MXd)^2}
=-\Bigl[\tfrac12\,K_1+\zeta\,\kappa\Bigr]=S_2.
\]
Since $S_2$ is quadratic in $\zeta$ and the integration
\eqref{eq:integrate} never raises the degree, $c_{(\MXd)^2}$ is
quadratic as well.
The subtraction that moment
matching performs order by order is made once, in Proposition
\ref{prop:K}, for every order at once.
\subsection{The degree structure of the prefactors}
The computation of Section \ref{sec:example} repeats at every order:
each term in the source is a product of a cumulant and prefactors of smaller
products, and the bookkeeping of the powers of $\theta$ and $\zeta$ is
the same at every step.  The result is the degree structure of the BG
expansion.
\begin{prop}\label{prop:degree}
For every product $\btau$ of $j$ trees of total order $\ell$,
\beq
c_{\btau}=(-\theta)^{j-1}\,q_{\btau}(\zeta,\theta),
\label{eq:structure}
\eeq
where $q_{\btau}$ is a polynomial of degree exactly $\ell$ in $\zeta$.
Moreover, the coefficient of $\zeta^\ell$ in
$q_{\btau}$ is positive.
\end{prop}
\begin{proof}
The proof proceeds by induction on $j$.  For $j=1$,
$c_\tau=w_\tau\,\cB[r_\tau]$; since $w_\tau>0$ and the coefficient of
$a^\ell$ in $r_\tau(a)$ is positive, the statement holds for $j=1$.
Suppose the statement holds for products of fewer than $j$ trees, and
consider a product $\btau$ in $\tilde \Sigma_j$ of total order $\ell$.
By Corollary \ref{cor:determined}, $c_{\btau}$ is obtained from the
$\btau$-part of the source \eqref{eq:sourceexpanded} by the
recursion \eqref{eq:integrate}.  

Write $\btau=\tau_1\,\btau_2$, with
$\tau_1$ a single tree of order $\ell_1$ and $\btau_2$ the product of
the remaining $j-1$ trees.  This splitting contributes to the $m=1$
terms of \eqref{eq:sourceexpanded}, up to a positive combinatorial
factor,
\[
-\Bigl[\,K_1\;\p_\zeta c_{\tau_1}\,\p_\zeta c_{\btau_2}
+\p_\zeta K_1\,\bigl(c_{\tau_1}\,\p_\zeta c_{\btau_2}
+c_{\btau_2}\,\p_\zeta c_{\tau_1}\bigr)\Bigr] ,
\]
with $K_1=\tfrac12\,(\zeta^2-\zeta-\theta)$ and $\p_\zeta K_1=\kappa$.
By the inductive hypothesis, $c_{\tau_1}$ has degree $\ell_1$ and
$c_{\btau_2}=(-\theta)^{\,j-2}\,q_{\btau_2}$ with
$\deg_\zeta q_{\btau_2}=\ell-\ell_1$, both with positive leading
coefficients, so this contribution is $-(-\theta)^{\,j-2}$ times a
polynomial of degree $\ell$ in $\zeta$ with positive leading
coefficient.

The remaining contributions to the $\btau$-part of
\eqref{eq:sourceexpanded} have the same form.  The terms with index $m$
are products of $m+1$ factors $\tilde\Sigma$ or $\p_\zeta\tilde\Sigma$,
so their $\btau$-part is a sum over splittings
$\btau=\btau_1\cdots\btau_{m+1}$ into sub-products of
$i_1,\dots,i_{m+1}$ trees, $i_1+\cdots+i_{m+1}=j$, of the terms
\[
-K_m\;\p_\zeta c_{\btau_1}\,\p_\zeta c_{\btau_2}\,
c_{\btau_3}\cdots c_{\btau_{m+1}}
\qquad\text{and}\qquad
-\p_\zeta K_m\;\p_\zeta c_{\btau_1}\,
c_{\btau_2}\cdots c_{\btau_{m+1}} ,
\]
again with positive combinatorial factors.  By \eqref{eq:Km},
$K_m=(-\theta)^{m-1}\bigl(\tfrac12\zeta^2+\cdots\bigr)$ and
$\p_\zeta K_m=(-\theta)^{m-1}\bigl(\zeta+\cdots\bigr)$, and by the
inductive hypothesis $c_{\btau_t}=(-\theta)^{\,i_t-1}q_{\btau_t}$, with
$q_{\btau_t}$ of degree the order $\ell_t\ge1$ of $\btau_t$ and positive
leading coefficient.  Each term above is therefore $-(-\theta)^{\,j-2}$
times a polynomial of degree $\ell$ in $\zeta$ with positive leading
coefficient: the exponents of $-\theta$ add up to
$(m-1)+\sum_t(i_t-1)=j-2$, the orders $\ell_t$ add up to $\ell$, and each
$\p_\zeta$ lowers the degree by one, which $K_m$ or $\p_\zeta K_m$
restores.  Summing over $m$ and over splittings, the $\btau$-part of
the source is $-(-\theta)^{\,j-2}$ times a polynomial of degree exactly
$\ell$ in $\zeta$ with positive leading coefficient.  Finally, the
integration \eqref{eq:integrate} supplies one factor of $\theta$ and
never raises the degree in $\zeta$, giving \eqref{eq:structure}.
\end{proof}
\noindent
With the substitutions $\zeta \mapsto \frac k M+\frac12$ and $\theta \mapsto \frac 1M$, the following corollary is immediate.
\begin{cor}\label{cor:degree}
In the smile expansion \eqref{eq:BG}, for every $\ell\ge1$, each
prefactor in $a_\ell(k)$, of a single tree or of a product of trees, is
a polynomial in $k$ of degree exactly $\ell$.  
\end{cor}
\begin{remark}\label{rem:logmechanism}
The mysterious cancellation is performed once and for all by the logarithm
$K=\log G$.
The recursion \eqref{eq:recursion} produces each
prefactor without generating higher-degree terms: the cumulants are quadratic in
$\zeta$, so at order $\epsilon^\ell$ the source term, and with it every
prefactor, has degree exactly $\ell$, while each additional tree in a
product costs one power of $\theta$ (Proposition \ref{prop:degree}).
\end{remark}
\begin{remark}\label{rem:generic}
Beyond the closed form \eqref{eq:Kclosed}, which uses that $\lambda_a$ is
quadratic, the recursion uses nothing about diamond trees except the
overall factor $\lambda_a$ that each carries in the CGF: Theorem
\ref{thm:recursion} holds for arbitrary
polynomials $r_\tau$.  
\end{remark}
\section*{Declarations}
\subsection*{AI use}
The authors made extensive use of generative artificial
intelligence tools (Claude Opus 5 and Claude Fable 5.1)
throughout the preparation of this manuscript. These tools were used to
assist with writing and editing the manuscript, exploring mathematical
ideas, proposing proof strategies, checking intermediate arguments, and
suggesting improvements to existing proofs. The AI tools were used solely
as assistants; all mathematical arguments were developed, critically
assessed, and independently verified by the authors. The authors assume
responsibility for all content.
\bibliographystyle{alpha}
\bibliography{MagicStrikes}

\clearpage
\appendix
\section{The BG expansion through order four}\label{app:BG4}
We give all coefficients of the total implied variance expansion
\eqref{eq:BG} through order four. We
use the variables $(\zeta,\theta)$ of Lemma~\ref{lem:Hermite}; substituting
$\zeta=\tfrac12+k/M$ and $\theta=1/M$ gives the expansion in $k$ and $M$.  Higher order trees and prefactors through order six are provided in the file \texttt{bg\_coefficients\_order\_6.txt}.
Consistent with Proposition \ref{prop:degree}, the prefactor of a product $\btau$ of $j$ trees of total order $\ell$ has the form $c_{\btau}=(-\theta)^{j-1}\,q_{\btau}(\zeta,\theta)$, with $q_{\btau}$ of degree exactly $\ell$ in $\zeta$.


\vspace{-25pt}
\begin{align}\
\intertext{\textbf{Order 1}}
 a_{1}(\zeta,\theta) &=\zeta\,\MXd.
\intertext{\textbf{Order 2}}
  a_{2}(\zeta,\theta) &=
\left(\zeta^{2} - \theta\right)\,\MXdXd+
\frac{1}{4}\left(\zeta^{2} - \zeta - \theta\right)\,\MMd
+\frac{1}{4}\,\theta\left(- 5\,\zeta^{2} + 3\,\zeta + 3\,\theta\right)\,\bigl(\MXd\bigr)^{2}.
\intertext{\textbf{Order 3}}
a_{3}(\zeta,\theta) &=
 \zeta\left(\zeta^{2} - 3\,\theta\right)\,\MXdXdXd
 +\frac{1}{2}\left(\zeta^{3} - \zeta^{2} - 3\,\zeta\,\theta + \theta\right)\,\left[\MXdMd +\frac{1}{2}\,\MMdXd\right]\\
&\qquad 
+\frac{1}{2}\,\theta\left(- 8\,\zeta^{3} + 5\,\zeta^{2} + 14\,\zeta\,\theta - 3\,\theta\right)\,\MXdXd\,\MXd\\
&\qquad
+\frac{1}{8}\,\theta\left(- 8\,\zeta^{3} + 10\,\zeta^{2} + 14\,\zeta\,\theta - 3\,\zeta - 6\,\theta\right)\,\MMd\,\MXd\\
&\qquad +\frac{1}{8}\,\theta^{2}\left(26\,\zeta^{3} - 24\,\zeta^{2} - 32\,\zeta\,\theta + 5\,\zeta + 10\,\theta\right)\,\bigl(\MXd\bigr)^{3}.
\intertext{\textbf{Order 4}}
 a_{4}(\zeta,\theta) &=\left(\zeta^{4} - 6\,\zeta^{2}\,\theta + 3\,\theta^{2}\right)\,\MXdXdXdXd
 +\frac{1}{8}\left(
  \zeta^{4} - 2\,\zeta^{3} - 6\,\zeta^{2}\,\theta + \zeta^{2} + 6\,\zeta\,\theta 
  + 3\,\theta^{2} - \theta
\right)\,\MMdMd\\
&\qquad+\frac{1}{2}\left(\zeta^{4} - \zeta^{3} - 6\,\zeta^{2}\,\theta + 3\,\zeta\,\theta + 3\,\theta^{2}\right)\,\left[ \MXdXdMd 
+\frac{1}{2}\,\MXdMXdd 
  +\MXdMdXd +\frac{1}{2}\,\MMdXdXd
\right]\\
&\qquad
+\frac{1}{2}\,\theta\left(- 6\,\zeta^{4} + 4\,\zeta^{3} + 21\,\zeta^{2}\,\theta - 7\,\zeta\,\theta - 7\,\theta^{2}\right)\,\bigl(\MXdXd\bigr)^{2}\\
&\qquad
+\frac{1}{64}\,\theta\left(
  - 12\,\zeta^{4} + 24\,\zeta^{3} + 42\,\zeta^{2}\,\theta - 15\,\zeta^{2} - 42\,\zeta\,\theta 
  + 3\,\zeta - 14\,\theta^{2} + 9\,\theta
\right)\,\bigl(\MMd\bigr)^{2}
\\
&\qquad
+\frac{1}{8}\,\theta\left(
  - 12\,\zeta^{4} + 16\,\zeta^{3} + 42\,\zeta^{2}\,\theta - 5\,\zeta^{2} - 28\,\zeta\,\theta 
  - 14\,\theta^{2} + 3\,\theta
\right)\,\MMd\,\MXdXd
\\
&\qquad
+\frac{1}{2}\,\theta\left(- 11\,\zeta^{4} + 7\,\zeta^{3} + 42\,\zeta^{2}\,\theta - 15\,\zeta\,\theta - 15\,\theta^{2}\right)\,\MXdXdXd\,\MXd
\\
&\qquad
+\frac{1}{4}\,\theta\,
\left(
\begin{aligned}
   &- 11\,\zeta^{4} + 15\,\zeta^{3} + 42\,\zeta^{2}\,\theta - 5\,\zeta^{2}\\
   &\qquad - 29\,\zeta\,\theta - 15\,\theta^{2} + 3\,\theta
\end{aligned}
 \right)\,
\left[\MXdMd +\frac{1}{2}\,\MMdXd\right]\,\MXd
 \\
&\qquad
+\frac{1}{8}\,\theta^{2}\left(
\begin{aligned}
  &144\,\zeta^{4} - 146\,\zeta^{3} - 378\,\zeta^{2}\,\theta + 35\,\zeta^{2}\\
  &\qquad + 198\,\zeta\,\theta + 102\,\theta^{2} - 15\,\theta
\end{aligned}
\right)\,\MXdXd\,\bigl(\MXd\bigr)^{2}\\
&\qquad
+\frac{1}{32}\,\theta^{2}\left(
\begin{aligned}
  &144\,\zeta^{4} - 224\,\zeta^{3} - 378\,\zeta^{2}\,\theta + 107\,\zeta^{2}\\
  &\qquad + 294\,\zeta\,\theta - 15\,\zeta + 102\,\theta^{2} - 45\,\theta
\end{aligned}
\right)\,\MMd\,\bigl(\MXd\bigr)^{2}\\
&\qquad +\frac{1}{64}\,\theta^{3}\left(
\begin{aligned}
  &- 692\,\zeta^{4} + 864\,\zeta^{3} + 1390\,\zeta^{2}\,\theta - 327\,\zeta^{2}\\
  &\qquad - 878\,\zeta\,\theta + 35\,\zeta - 306\,\theta^{2} + 105\,\theta
\end{aligned}
\right)\,\bigl(\MXd\bigr)^{4}.
\end{align}

\end{document}